\documentclass[letterpaper, 10 pt, conference]{ieeeconf}  

\IEEEoverridecommandlockouts                              
\usepackage{amsmath,amssymb,amsfonts,bbm}
\usepackage{graphicx}
\usepackage{algorithm2e}
\usepackage{textcomp}
\usepackage{xcolor}
\usepackage{hyperref}
\DeclareMathOperator*{\argmax}{arg\,max}
\DeclareMathOperator*{\argmin}{arg\,min}
\DeclareMathOperator*{\E}{\mathbb{E}}
\newtheorem{theorem}{Theorem}

\newcommand{\online}[1]{#1}
\renewcommand{\online}[1]{}

\title{\LARGE \bf
Tractable Defense Against Advanced Persistent Threats in Networks}

\author{Brandon Collins, Keith Paarporn, Shouhuai Xu, Philip N. Brown
\thanks{Research was sponsored by the Air Force Office of Scientific Research under award number FA9550-23-1-0171, by the Army Research Office under grant number W911NF-25-1-0239, by NASA under grant number 80NSSC25M7102, and by DoD UC2 program. The views and conclusions contained in this document are those of the authors and should not be interpreted as representing the official policies, either expressed or implied, of the AFOSR, Army Research Office, DoD, NASA, or the U.S. Government. The U.S. Government is authorized to reproduce and distribute reprints for Government purposes notwithstanding any copyright notation herein.}
\thanks{Brandon Collins, Keith Paarporn, and Shouhuai Xu are with the University of Colorado Colorado Springs, Colorado Springs, CO, 80918 USA (email: {\tt bcollin3@uccs.edu})}%
\thanks{Philip N. Brown is with the University of Colorado Colorado Springs, Colorado Springs, CO, 80918 USA (email: {\tt pbrown2@uccs.edu}) and Politecnico di Torino, Italy.}%
}

\begin{document}

\maketitle
\thispagestyle{empty}
\pagestyle{empty}

\begin{abstract}
Recently, the theory of Boolean Dynamical Systems was proposed to study the decision theory surrounding the defense of computer networks against Advanced Persistent Threats (APTs).
Boolean Dynamical Systems naturally capture four first principle primitives of APTs: the stealthy nature of attacks, limited and noisy information from automated systems like intrusion detection systems, lateral movement after the attacker penetrates into the network, and the defender's ability to secure a subset of computers at any time at the loss of resources such as system uptime.
Currently, doing optimal/heuristic control in a computationally tractable manner is not possible because the emergent value function is computationally intractable (with respect to the network size).
To resolve this, we propose a mean-field analysis inspired heuristic value function.
We prove that our proposed heuristic is based on an exact computation of the value function under the assumption that the underlying state estimate distribution maximizes entropy.
We numerically evaluate the quality of our heuristic as parameterized by the degree to which the entropy assumptions are violated. 
\end{abstract}

\section{Introduction}
In recent years, the prevalence of Advanced Persistent Threats (APT) continues to increase.
Broadly, the APT can be described as follows: a resourceful and sophisticated adversary leverages stealthy attacks against a system operator's computer network.
Once the adversary gains a foothold in the network, they conduct \textit{lateral movement}, where an adversary spreads their control 
and \textit{privilege escalation} to increase their ability to exert control over the network.
Once they have established a sufficient foothold, they can often lie dormant, either passively conducting attacks or waiting for an impactful moment to strike.
This process from an adversary scouting the network using \textit{reconnaissance} to delivering a \textit{payload} is well studied in cybersecurity, and many frameworks have been proposed to understand the significant moments in a cyber attack, such as MITRE ATT\&CK \cite{al2024mitre} and Cyber Kill Chain \cite{senatereport2014}.
Given the massive impact that APT-related breaches can have, it is of critical importance that computer network operators (or defenders) have robust tools and decision-making processes to handle these threats.

The APT problem domain is characterized by three unique features.
First, the defender 
has \textit{low-information} meaning that the system operator cannot directly detect adversarial activities in their network, often due to the exploitation of \textit{zero-day} vulnerabilities which are considered to be unknown vulnerabilities and therefore impossible to detect.
Second, the attacker moves freely in a victim network after establishing a footprint, called lateral movement.
Third, automated cyber defenses are widely deployed by network operators (i.e., defenders).
Chief among these defense mechanisms is the \textit{intrusion detection system} (IDS), which monitors either network traffic in a network or the behavior of computers.
These IDS sensors automatically detect suspicious activities and functions as a noisy signal of the true  state of the network.

Recently, it has been recognized that Boolean Dynamical Systems (BDSes) can describe these three characteristics
\cite{kazeminajafabadi2024optimal,kazeminajafabadi2024optimal_journal,collins2025efficient}.
To model APT attacks in a network,
we suppose each node in the network is a computer, and the network connections define the permitted connections between computers, which is an accepted assumption in theoretic cybersecurity research (cf. \cite{XuBookChapterCD2019,XuMTD2020,XuTDSC2012} and the references therein).
Each computer has a boolean state that indicates if the computer is infected or not.
The state dynamics are then defined analogously to the compartmental epidemic equations \cite{pare2020modeling}, where an infected computer can spread the infection to clean connected computers.
The APT-based BDS models differ from the epidemic spreading process in the implementation of how a node recovers; in an APT setting, we say the network operator \textit{cleans} a computer.
Thus, we consider the cleaning of a computer to be a control input, where the cleaning action has an associated cost in lost resources (e.g. lost uptime)

From a control theoretic point of view, this problem can be broken in two parts: (i) an estimation problem
of converting the noisy IDS signals, historic cleaning actions and knowledge of the network structure into an estimate of the current state; and (ii) the control problem of deciding what computers to clean, balancing the cost of cleaning a computer with the security benefits of doing so.
Solving both of these issues is further complicated by the exponential size of the state and control space.
In particular, a network of size $n$ has $2^n$ states, and the network operator can clean any computer at any time, meaning the action space at a single timestep also has $2^n$ possibilities.
Thus, even though control-theoretic techniques have been utilized to solve both problems (as in \cite{kazeminajafabadi2024optimal_journal}), the resulting solutions are computationally intractable, meaning that solutions cannot be realistically computed for large networks.
Here, we use the term \textit{computationally tractable} to differentiate algorithms that run in polynomial versus ones that are \textit{intractable} or run in exponential time, both with respect to the size of the network.
For cybersecurity, tractable algorithms are of critical importance, as solutions need to be computed on large networks in real time.

Thus, in this paper, we aim to provide heuristic and tractable control solutions. 
On the estimation side, the Boolean Kalman Filter (BKF) \cite{braga2011optimal} provides intractable but provably-optimal estimate distributions.
Our recent work in~\cite{collins2025efficient} has shown that the BKF algorithm can be tractably computed under the assumption that the relevant state estimation distribution maximizes entropy in a certain constrained set of possible distributions.
In the present paper, we build on that work towards solving the control portion of the problem. 
We recognize that before control approaches can be investigated, it is first necessary to evaluate different controls and assess which one is \textit{optimal}, or the best control to choose among all options.
Formally, we define a \textit{value function} that provides a numeric evaluation of any possible control.

The core issue that this paper seeks to address is that the emergent value function in this problem domain is defined via (exponentially) large matrix operations that are computationally intractable, rendering existing optimal solutions unusable in practical settings.
In this work, we provide a computationally tractable heuristic value function that enables future work to assess the best approach for optimal control in this problem domain.
To this end, our contributions are as follows:
\begin{enumerate}
    \item We prove that the expected reward for a given state-action pair can be tractably computed.
    \item We prove that the expected reward with respect to a belief distribution over the state space and a selected action can be tractably computed in a single timestep, given that the distribution maximizes entropy.
    \item We leverage the above theoretical results to propose a novel heuristic and tractable value function that evaluates any finite length control sequence.
    \item Using simulations, we numerically evaluate the performance of the proposed heuristic value function with respect to the entropy of the initial distribution.
\end{enumerate}

\section{Model}
\subsection{Model}
We model the spread of computer viruses through a network of computers.
Let the set of computers be $N=\{1,2,\dots,n\}$ for some positive integer $n$.
To represent the state of the network at a time $t\in \mathbb{N}$, we use a vector $x_t\in \{0,1\}^n$ where $x_{t,l}$ represents the state of the computer $l$ at time $t$, where $1\leq l \leq n$.
If $x_{t,l}=1$, we say that computer $l$ is infected at time $t$, and similarly if $x_{t,l}=0$ then computer $l$ is \textit{clean} or \textit{secure}. 
Let $X=[x^1,x^2,\ldots, x^{2^n}]$ be an $n \times 2^n$ matrix where each column $x^i\in\{0,1\}^n$ represents a unique state of the network.
Together, all of the columns $x^i$ of $X$ include all possible states of the network in $\{0,1\}^n$.
Occasionally, it is convenient to abuse notations somewhat and treat $X$ as a set, namely $X=\{x^1,x^2,\dots,x^{2^n}\}$.
We denote a distribution over all $2^n$ states as $\Pi\in\Delta(X)$, where $\Delta(X)=\{\Pi\in[0,1]^{2^n}\mid \sum_{i\in N} \Pi_i =1 \}$ is the space of all distributions over $X$.
Further, suppose that the network enforces a security policy that only permits certain pairs of computers to communicate with each other, effectively leading to a network structure.
We denote this policy as a set $E$, and we say that computer $l$ can propagate an attack to computer $k$ if $(l,k)\in E$.
This communication is directed, so $(l,k)\in E$ does not imply that $(k,l)\in E$.
Additionally, let $\vec{1}$ indicate an appropriately sized column vector of 1's.

To model the spread of APT attacks in the network, infected computers may propagate over permitted communication links.
In particular, if the computer $l$ is infected ($x_{t,l}=1$) then at each timestep it has the probability $\rho_{lk}\in[0,1]$ to propagate the attack to the neighbor node $k$.
That is, if $(l,k)\in E$ and $x_{t,l}=1$, the infection (i.e., attack) will spread through this connection with probability $\rho_{lk}$ at each time step.
For mathematical convenience, we define $\rho_{lk}=0$ anytime $(l,k)\notin E$, and $\rho_{ll}\notin E$ for all nodes $l\in N$ (that is, there are no self-edges).
Additionally, we define the \textit{in-neighbors} of a node $l$ as $D_l=\{k\in N\mid (k,l)\in E\}$.

The APT scenario is typically thought to model \textit{low} or \textit{no} information settings because of their stealthy behavior.
That is, the defender has limited information about the true state of their network at time $t$, or $x_t$.
To capture the situation where the defender has a realistic amount of information, we assume that the defender has IDS enabled on each computer.
These IDS systems produce periodic and noisy observations of the true state of the network at any time $t$, or $x_t$.
Particularly, we assume they produce observation $y_t\in \{0,1\}^n$ at each time step, where $y_{t,l}=1$ indicates that the IDS flagged computer $l$ being infected at time $t$, and similarly $y_{t,l}=0$ indicates that the computer is considered clean by the IDS.
The following equation gives the rates at which the IDS gives false negatives and false positives, parameterized by $p,q\in [0,1]$:
\begin{equation}
y_{t,l}=\begin{cases}
        1 \mbox{ if }x_{t,l}=1 & \mbox{with probability }p\\
        0 \mbox{ if }x_{t,l}=1 & \mbox{with probability }1-p\\
        0 \mbox{ if }x_{t,l}=0 & \mbox{with probability } q\\
        1 \mbox{ if }x_{t,l}=0 & \mbox{with probability }1-q\\
    \end{cases}
\end{equation}

Finally, given this noisy information, the defender may choose to clean any subset of computers at each time step.
We notate the defender's actions as $a_t\in\{0,1\}^n$ where $a_{t,l}=1$ indicates the defender's actions.
If a defender attempts to clean an infected computer, we assume that there is a chance the cleaning will fail $\alpha\in[0,1]$, which means that the cleaning operation fails to change the state of a compromised computer to a secure computer, with probability $\alpha$.

Given this, we can now give the true dynamics of the system.
That is, the probability of a node $l$ being infected at time $t+1$ or in $x_{t+1}$ is given by 
\begin{equation}\label{eq:xt dynamics}
\begin{aligned}
    \Pr(x_{t+1,l}=1\mid x_t,a_t)&=
    (1+(\alpha-1) a_{t,l})\bigg(x_{t,l}
    +\\
    &(1-x_{t,l})\bigg[1-\prod_{k\in D_l}(1-\rho_{kl}x_{t,k})\bigg]  \bigg).
\end{aligned}
\end{equation}
The above equation encodes the network spreading dynamics given the network structure and current state.
When combined with the past sequences of observations $\vec{y}=[y_1,y_2,\dots,y_{t}]$ and cleaning inputs $\vec{a}=[a_1,a_2,\dots,a_{t}]$, it has been shown that the BKF optimally estimates the true state $x_t$ \cite{braga2011optimal}.

Thus, we can leverage \eqref{eq:xt dynamics} to produce estimates of the true state $x_t$.
To evaluate any state estimate technique, we use the least squares criterion, given as follows:
\begin{equation}
    \hat{x}^*_t\in \argmin_{\hat{x}_t\in \Psi} \mathbb{E}(||x_t-\hat{x}_t(a_{0:t-1},y_{1:t})||^2_2\mid a_{0:t-1},y_{1:t}),
\end{equation}
where $\Psi$ is the space of all estimators.
It is worth emphasizing that the BKF solution is computationally intractable on large networks.
Notably though, in addition to producing an estimate $\hat{x}_t\approx x_t$, in doing so it keeps track of a running distribution $\Pi_t\in \Delta(X)$.
Using this distribution, we can calculate the expected value of future actions.
Before showing that, we develop the transition matrix that enables the calculation of $\Pi_{t+1}$ from $\Pi_t$ given the knowledge of the network structure \eqref{eq:xt dynamics}.

First, we rewrite \eqref{eq:xt dynamics} as the probability node $l$ gets infected in the next time step given that the current state is $x^j\in X$ and some action $a_{t-1,l}$ is taken:
\begin{equation}
\begin{aligned}
    \eta^j_l(a_{t-1,l})=&\bigg(1+a_{t-1,l}(\alpha-1)\bigg)\\
    &\quad\bigg(x^j_l+(1-x^j_l)\bigg[1-\prod_{r\in D_l}(1-\rho_{rl}x^j_r)\bigg]\bigg).
\end{aligned}
\end{equation}
We leverage $\eta^j_l$ to define the column-stochastic matrix $M_t$, which is the transition matrix that defines the hidden Markov chain over the $2^n$ states.
Specifically, the entry $(M_t)_{ij}$ is the probability that the state $x^j$ will transition to $x^i$, namely:
\begin{equation}\label{eq:def Mt}
\begin{aligned}
    (M_t)_{ij}&=\Pr(x_t=x^i\mid x_{t-1}=x^j,a_{t-1}) \\
    &=\prod^n_{l=1}\bigg(\eta^j_l x^i_{t,l}+(1-\eta^j_l)(1-x^i_{t,l}) \bigg).
\end{aligned}
\end{equation}
Using this transition matrix, we can now calculate expected future distributions given an action sequence.
Using this, we want to calculate the \textit{expected reward} of a given action sequence.
To do this, we introduce a \textit{reward matrix}, given by:
\begin{equation}
    R_{ij}(a)=-\sum_{l\in \{1,2,\dots,n\}}(x_l^i+ca_l).
\end{equation}
Similarly to the transition matrix, $R_{ij}(a)$ is the reward received when selecting action $a$ in state $x^j$ and moving to a new state $x^i$.
If the network operator cleans the node $l$, this operation incurs an immediate cost $c>0$, but they will receive improved future rewards because the cleaned computer $l$ is less compromised in expectation, reducing the penalty for compromised computers in the future.
Further, if the computer $l$ is cleaned, then its chances of spreading infection in the future are greatly reduced, creating a balancing act between the short term cleaning cost and the longer term network effects of a cleaner network.
To aggregate the reward over time, we use a \textit{discounted value function}, given by:
\begin{equation} \label{eq:exact value}
 V(\vec{a})=\sum_{t=0}^T \bigg(\gamma^t\sum_{x^j\in X}\Pi^j_t R(a_t)\odot M(a_t)^T  \vec{1} \bigg)
\end{equation}
where the \textit{discount factor} $\gamma\in (0,1)$ balances the short term reward with discounted future rewards.
The term in the sum can be regarded as the expected reward given the current distribution $\Pi_t$ and the transition probabilities incurred by the action $a_t$.
Ultimately, the goal of this line of work is an algorithm that can produce optimal or near-optimal control sequences $\vec{a}^*$ for any initial distribution $\Pi_0$ that maximizes the discounted reward function.
Formally, this objective is given by:
\begin{equation}\label{eq:objective}
    a^*(\Pi_0)\in \argmax  V(\vec{a}).
\end{equation}
This problem, where the current state is unknown but noisy information is available, and the operator must make decisions under uncertainty with respect to the present state, is known as a \textit{Partially Observable Markov Decision Problem} (POMDP).

\section{Estimating The Value Function}
\subsection{A Theoretical Approach}

Even in the case where there is no exponential state/action space,
POMDP's are computationally difficult to solve \cite{papadimitriou1987complexity}.
One thread of research to resolve this is to solve POMDPs heuristically \cite{hauskrecht2000value} in a tractable fashion.
Unfortunately, even these heuristics are built on the assumption that the value function itself can be computed in a computationally tractable manner.
Currently, the value function \eqref{eq:exact value} requires matrix computations of order $2^n$, which means it is itself an intractable calculation.
This implies that any heuristic algorithm that attempts to evaluate the value function itself will become computationally intractable.
This is the core motivation for the present work, creating a tractable estimate of the value function to enable the use of heuristic algorithms to solve this problem.

We break evaluating the inner term of \eqref{eq:exact value} into two parts, first we consider the situation where we want to evaluate the reward function given an action $a$ and a current state $x^j\in X$, given by:
\begin{equation}
    L^j(a)=\E_{i\sim M_{j}(a)}(R_{ij}(a))=\sum_{i\in X}(M_{ij}(a)R_{ij}(a)).
\end{equation}
This computation is intractable with respect to the network size, because it requires a $2^n\times 2^n$ matrix computation as well as summing $2^n$ terms.

We now give our first result, which gives a method for a tractable calculation of the reward, given that the source state $j$ and the action $a$ are known with certainty.
\begin{theorem} \label{Thm:expected reward}
    The expected reward given the present state is $x^j$ and some action $a\in \{0,1\}^n$ can be computed as:
    \begin{equation}\label{eq:L def}
       L^j(a)=-\sum_{k\in N}(\eta^j_k+ca_k).
    \end{equation}
\end{theorem}
\begin{proof}
Beginning with equation

\begin{equation} \label{proof:eq:exp R pt1}
\begin{aligned}
    &\sum_{i\in X}\bigg(M_{ij}(a)R_{ij}(a)\bigg)= \sum_{i\in X}\bigg(M_{ij}(a)(-\sum_{k\in N}(x_k^i+ca_k)\bigg) \\
    &=\sum_{i\in X}\bigg(M_{ij}(a)\sum_{k\in N}(-x_k^i)\bigg)        +\sum_{i\in X}\bigg(M_{ij}(a)\sum_{k\in N}(-ca_k)\bigg)\\
    &=\sum_{i\in X}\bigg(M_{ij}(a)\sum_{k\in N}(-x_k^i)\bigg)+\sum_{k\in N}(-ca_k)
\end{aligned}
\end{equation}
where the first equality is the distributive law, the second is a factoring of $\sum_{k\in N}(-ca_k)$ as it does not depend on $i$, and acknowledges that $\sum_{i\in X}M_{ij}(a)$ is a sum over a distribution.
Next, we focus on the left term in the final equality:
\begin{equation}
\begin{aligned}
    &\sum_{i\in X}\bigg(M_{ij}(a)\sum_{k\in N}(-x_k^i)\bigg) \\
    &=-\sum_{i\in X}\sum_{k\in N}\bigg(x_k^i \prod_{l\in N}\big(\eta^j_l x^i_{l}+(1-\eta^j_l)(1-x^i_{l}) \big)\bigg) \\
    &=-\sum_{k\in N}\sum_{i\in X}\bigg(x_k^i \prod_{l\in N}\big(\eta^j_l x^i_{l}+(1-\eta^j_l)(1-x^i_{l}) \big)\bigg) \\
    &=-\sum_{k\in N}\sum_{i\in X^1_k}\bigg( \prod_{l\in N}\big(\eta^j_l x^i_{l}+(1-\eta^j_l)(1-x^i_{l}) \big)\bigg) \\
    &=-\sum_{k\in N}\sum_{i\in X^1_k}\bigg(\eta^j_k \prod_{l\in N\setminus \{k\}}\big(\eta^j_l x^i_{l}+(1-\eta^j_l)(1-x^i_{l}) \big)\bigg) \\
    &=-\sum_{k\in N} \eta^j_k \sum_{i\in X^1_k}\bigg( \prod_{l\in N\setminus \{k\}}\big(\eta^j_l x^i_{l}+(1-\eta^j_l)(1-x^i_{l}) \big)\bigg) \\
    &=-\sum_{k\in N} \eta^j_k. 
\end{aligned}
\end{equation}
The sums are justified as follows: the first equality follows by definition of $M_{ij}(a)$, and distributes the product into the sum over $N$; the third swaps the sums as there is only one term inside the inner sum; the fourth introduces set $X^1_k=\{i\in X\mid x^i_k=1\}$, recognizing $x^i_k=1$ is required for the term inside to be non-zero; the fifth factors out $\eta^j_k$ as $i\in X^1_k$ implies $\eta^j_k$ is in the product term; the sixth factors $\eta^j_k$ out of the sum over $X^1_k$, recognizing it has no dependence on variable $i$; and the final equality follows as the sum $X^1_k$ is now a sum over a distribution.
The proof concludes by replacing the above equality into
\eqref{proof:eq:exp R pt1}, obtaining 
\begin{equation}
    \sum_{i\in X}\bigg(M_{ij}(a)R_{ij}(a)\bigg)=-\sum_{k\in N}(\eta^j_k+ca_k)
\end{equation}
as desired.
\end{proof}

The above result enables us to compute the expected value of an action $a$ given a state $x^j$.
This computation is essential for the perfect information version of the problem.
However, in the incomplete information version, we must calculate the expected reward of an action given some distribution $\Pi\in \Delta(X)$; in practice, $\Pi$ is generated from the BKF algorithm.
Unfortunately, $\Pi$ is a probability distribution with $2^n$ elements, so direct computations of expectations are generally not tractable.
To resolve this, instead, we base our current state estimate on a \textit{marginal probability} vector $P_t\in [0,1]^n$.
We term $P_t$ the marginal probability, as given any distribution $\Pi_t$, we can calculate it via $X\Pi_t=P_t$.
Assuming that we do not have access to $\Pi_t$ (or the computational ability to compute $X\Pi_t)$, we can instead before updates on $P_t$ based on the so called \textit{mean field analysis} equation:
\begin{equation}\label{eq:MFA}
\begin{aligned}
P_{t,l}&=(1+(\alpha-1)a_{t,l}) \\
&\bigg( P_{t-1,l}+(1-P_{t-1,l})\bigg[1-\prod_{k\in D_l} (1-\rho_{kl}P_{t-1,k})\bigg]\bigg).    
\end{aligned}
\end{equation}
Here, where $P_{t,l}$ is the probability that computer $l$ is compromised.
Notably, to generate this equation we simply replaced $x_t$ or the true state with the expected value that the node is actually compromised.
The critical difference between \eqref{eq:MFA} and \eqref{eq:xt dynamics} is that \eqref{eq:MFA} only needs to be evaluated $n$ times to compute the vector $P_t$, whereas \eqref{eq:xt dynamics} needs to be computed $2^n$ times for each state $x^i\in X$.

We say a distribution $\Pi$ is \textit{consistent} with information $P$ if $\Pi\in S(P_t)=\{\Pi'\in \Delta(X)\mid X\Pi'=P_t\}$.
Then, among the distributions that are consistent with $P$, we are interested in the distribution that maximizes the entropy.
Formally, this is given by:
\begin{equation}\label{eq:Pistar def}
    \Pi^*(P)\in\argmax_{\Pi\in S(P)} \sum_{i}-\Pi_i\ln \Pi_i.
\end{equation}
Previous work \cite{collins2025efficient} (Theorem 1) provides a closed-form representation of such entropy maximizing distributions,
\begin{equation}\label{eq:max entropy}
        \Pi^*(P)_i=\prod_{l\in V} ( P_lx^i_l+(1-x^i_l)(1-P_l) ).
\end{equation}
Using these entropy maximizing distributions, we now provide a result to evaluate the expected cost of an action given an entropy maximizing distribution $\Pi^*(P)$.
Specifically, we show that if the distribution $\Pi$ maximizes entropy among all distributions consistent with some $P$ (that is, $\Pi=\Pi^*(P)$), then the value function can be computed exactly in a tractable fashion.

\begin{theorem} \label{thm:expected value}
Let $a\in\{0,1\}^n$ be an action and $P_t\in[0,1]^n$ be the present information, the given distribution $\Pi^*(P_t)$ can be computed in polynomial time via
    \begin{equation} \label{eq:MFA value approx}
        \E_{j\sim \Pi^*(P_t)}\big(L^j(a)\big)=-\sum_{l\in N}(P_{t+1,l}+ca_l).
    \end{equation}
\end{theorem}
\begin{proof}
For some $P\in[0,1]^n$ and $a\in\{0,1\}^n$, we begin with the expectation of $L^j(a)$ with respect to the induced distribution $\Pi^*(P)$ :
\begin{equation}\label{proof:eq:L max entropy pt1}
\begin{aligned}
    &\E_{j\sim \Pi^*(P)}(L^j(a))=\sum_{j\in X}\Pi^*_j(P)L^j(a)\\
    &=\sum_{j\in X}\bigg(\Pi^*_j(P)\big(-\sum_{l\in N}\eta^j_l+ca_l\big)\bigg)\\
    &=-\sum_{l\in N}ca_l-\sum_{j\in X}\bigg(\Pi^*_j(P)\sum_{l\in N}\eta^j_l\bigg),
\end{aligned}
\end{equation}
where the first equality is the definition of expectation, the second follows from \eqref{eq:L def}, and the third as $ca_k$ can be factored out of the sum over $X$, leaving only a sum over a distribution remaining.
For compactness, we now focus on the second term. 
Letting $\tilde{a}_l=(1+a_{t-1,l}(\alpha-1)$, $\bar{P}_l=1-P_l$ and $\beta^i_l=\prod_{k\in D_l}(1-\rho_{kl}x^i_k)$, we have:
\begin{equation}\label{proof:eq:L max entropy pt2}
\begin{aligned}
    &-\sum_{j\in X}\bigg(\Pi^*_j(P)\sum_{l\in N}\eta^j_l\bigg)\\
    &=-\sum_{j\in X}\sum_{l\in N}\bigg(\tilde{a}_l\bigg(x^j_l+\bar{x}^j_l(1-\beta_l^j)\bigg)\prod_{k\in N} ( P_k x^j_k+\bar{x}^j_k\bar{P}_k )\bigg) \\
    &=-\sum_{l\in N} \tilde{a}_l\sum_{j\in X}\bigg(x^j_l\prod_{k\in N} ( P_k x^j_k+\bar{x}^j_k\bar{P}_k )+ \\
    &\qquad\qquad\qquad\qquad\bar{x}^j_l(1-\beta_l^j)\prod_{k\in N} ( P_k x^j_k+\bar{x}^j_k\bar{P}_k )\bigg) \\
    &=-\sum_{l\in N} \tilde{a}_l\bigg(\sum_{j\in X}\big( x^j_l\prod_{k\in N} ( P_k x^j_k+\bar{x}^j_k\bar{P}_k )\big)+ \\
    &\qquad\qquad\qquad\qquad\sum_{j\in X}\big(\bar{x}^j_l(1-\beta_l^j)\prod_{k\in N} ( P_k x^j_k+\bar{x}^j_k\bar{P}_k )\big)\bigg) \\
    &=-\sum_{l\in N} \tilde{a}_l\bigg(\sum_{j\in X^1_l}\big( P_l\prod_{k\in N\setminus l} ( P_k x^j_k+\bar{x}^j_k\bar{P}_k )\big)+ \\
    &\qquad\qquad\sum_{j\in X^0_l}\big((1-P_l)(1-\beta_l^j)\prod_{k\in N\setminus l} ( P_k x^j_k+\bar{x}^j_k\bar{P}_k )\big)\bigg) \\
    &=-\sum_{l\in N} \tilde{a}_l\bigg(P_l+(1-P_l)\big(\sum_{j\in X^0_l}\prod_{k\in N\setminus l} ( P_k x^j_k+\bar{x}^j_k\bar{P}_k ) \\ \\
    &\qquad\qquad-\sum_{j\in X^0_l}\prod_{k\in D_l}(1-\rho_{kl}x^i_k)\prod_{k\in N\setminus l} ( P_k x^j_k+\bar{x}^j_k\bar{P}_k )\big)\bigg) \\
    &=-\sum_{l\in N} \tilde{a}_l\bigg(P_l+(1-P_l)\big(1-\sum_{j\in X^0_l}\prod_{k\in D_l}(1-\rho_{kl}x^j_k) \\
    &\qquad\qquad\prod_{k\in N\setminus l} ( P_k x^j_k+\bar{x}^j_k\bar{P}_k )\big)\bigg). \\
\end{aligned}
\end{equation}
Each step can be described as follows: the first equality expands the definitions of $\Pi^*_j(P)$ and $\eta^j_l$, The second equality flips the sims of $N$ and $X$, factors out $\tilde{a}$ and distributes the product over $N$,
third distributes the sum, the fourth equality leverages changes the sum to be over $X^1_l=\{i\in X\mid x^i_k=1\}$ and $X^0_l=\{i\in X\mid x^i_k=0\}$ which true because in the cases $x^j$ is not in the appropriate set the sum term will be zero, additionally we can factor out a $P_l$ and a $1-P_l$ respectively as the sum sets now guarentee that term appears in the product.
In the fifth equality, we factor out $P_l,1-P_l$ and recognize that in $P_l$'s case the remaining sum is over a distribution, so it is equal to 1, additionally we distribute the term $1-\beta^j_l$ across the product.
In the final equality, notice $\sum_{j\in X^0_l}\prod_{k\in N\setminus l} ( P_k x^j_k+\bar{x}^j_k\bar{P}_k )=1$ as it is a sum over a distribution.

Now we focus on the final sum term over $X^0_l$ for some $l\in N$,
\begin{equation}
\begin{aligned}
    &\sum_{j\in X^0_l}\bigg(\prod_{k\in D_l}(1-\rho_{kl}x^j_k)\prod_{k\in N\setminus l} ( P_k x^j_k+\bar{x}^j_k\bar{P}_k )\bigg) \\
    &= \sum_{j\in X^0_l}\bigg(\prod_{k\in N\setminus l}(1-\rho_{kl}x^j_k) ( P_k x^j_k+\bar{x}^j_k\bar{P}_k )\bigg)\\
    &= \sum_{j\in X^0_l}\bigg(\prod_{k\in N\setminus l}x^j_kP_k(1-\rho_{kl})+(1-x^j_k)(1-P_k) \bigg) \\
\end{aligned}
\end{equation}
Where, in the first equality we combine the product terms, which is permissible because $(k,l)\notin E$ implies that $\rho_{kl}=0$, and there are no  self-loops $\rho_{ll}=0$. The second is simply distributing the terms in the sum, and leveraging the binary nature of the variable's values.

Now, we let $Y=X^0_l$ and select some $r\in N\setminus\{l\}$ and define $Y^1_r=\{x^i\in X^0_l\mid x^i_r=1\}$ and similarly
$Y^0_r=\{x^i\in X^0_l\mid x^i_r=0\}$.
For space, let $f(j,k,l)=x^j_kP_k(1-\rho_{kl})+(1-x^j_k)(1-P_k)$ which leads to the following decomposition:
\begin{equation}
\begin{aligned}
    &\sum_{j\in X^0_l}\bigg(\prod_{k\in N\setminus l}x^j_kP_k(1-\rho_{kl})+(1-x^j_k)(1-P_k) \bigg) \\
    &=\sum_{j\in Y^1_r}\bigg(\prod_{k\in N\setminus l}x^j_kP_k(1-\rho_{kl})+(1-x^j_k)(1-P_k) \bigg) \\
    &\quad+\sum_{j\in Y^0_r}\bigg(\prod_{k\in N\setminus l}x^j_kP_k(1-\rho_{kl})+(1-x^j_k)(1-P_k) \bigg) \\
    &=(P_k(1-\rho_{kl}))\sum_{j\in Y^1_r}\bigg(\prod_{k\in N\setminus \{l,r\}}f(j,k,l) \bigg) \\
    &+(1-P_k)\sum_{j\in Y^0_r}\bigg(\prod_{k\in N\setminus \{l,r\}}f(j,k,l) \bigg) \\
    &=((P_r-P_r\rho_{rl})+1-P_r)\sum_{j\in Y^1_r} \\ 
    &\quad \bigg(\prod_{k\in N\setminus \{l,r\}}x^j_kP_k(1-\rho_{kl})+(1-x^j_k)(1-P_r) \bigg) \\
    &=(1-P_r\rho_{rl})\sum_{j\in Y^1_r}\bigg(\prod_{k\in N\setminus \{l,r\}}f(j,k,l) \bigg).
\end{aligned}
\end{equation}
Now, we can set $Z=Y^1_r$ and select some $r'\in N\setminus\{l,r\}$ and repeat the above procedure.
Once we repeat it for all elements in $N$, we obtain the equality
\begin{equation}
\begin{aligned}
    &\sum_{j\in X^0_l}\bigg(\prod_{k\in N\setminus l}x^j_kP_k(1-\rho_{kl})+(1-x^j_k)(1-P_k) \bigg) \\
    &=\prod_{r\in N}(1-\rho_{rl} P_r).
\end{aligned}
\end{equation}
This can be substituted back into \eqref{proof:eq:L max entropy pt2} to obtain
\begin{equation}
\begin{aligned}
&-\sum_{j\in X}\bigg(\Pi^*_j(P)\sum_{l\in N}\eta^j_l\bigg)\\
&=-\sum_{l\in N} \tilde{a}_l\bigg(P_l+(1-P_l)\big(1-\prod_{r\in N}(1-\rho_{rl} P_r)\bigg)\\
&=-\sum_{l\in N} P_{t+1,l}
\end{aligned}
\end{equation}
which can be substituted back into \eqref{proof:eq:L max entropy pt1}
\begin{equation}
    \sum_{j\in X}\Pi^*_j(P)L^j(a)=-\sum_{l\in N}(ca_l+P_{t+1,l})
\end{equation}
as desired.
\end{proof}
The main purpose of this approach is that computing \eqref{eq:MFA}) for all nodes $l\in N$ runs in the worst case in $\mathcal{O}(n^2)$ (this worst occurs occurs on a complete graph).
In practice, this means that computing \eqref{eq:MFA value approx} can be done quickly on large graphs.

Interestingly, this computation has direct connections to mean field analysis (E.g. requiring a direct computation of \eqref{eq:MFA}) equations that are often used to study epidemic spreading processes \cite{pare2020modeling,paarporn2017networked,nowzari2016analysis}.
In particular, the equations that govern $P_{t+1,l}$ are exactly the same equations used to approximate the spread of an epidemic throughout the network.
There, they overcome the same curse of dimensionality by changing the focus from a $2^n$ distribution $\Pi$, and instead analyzing a more manageable marginal distribution vector $P\in[0,1]^n$.
The main difference between these approaches is that in biological networks the nodes represent living beings that recover from the disease naturally over time.
In our model, we regard the recovery process as the control input, where network operators (i.e., defenders) choose which nodes of the network to invest resources in securing.

The strength of this result is that our estimation is exact in the special case where the distribution $\Pi$ has $\Pi=\Pi^*(P)$ for some $P$.
However, in practice, we can only tractably compute $P_t$, so we will evaluate the value function \eqref{eq:exact value} under the assumption that the true estimate distribution $\Pi_t=\Pi^*(P_t)$.
There are two issues with this estimate that can cause the heuristic evaluation to drift from the value function.
The first issue occurs even in the optimistic setting where we suppose that the initial true distribution $\Pi_0$ maximizes entropy for our initial estimate $P_0$ (that is, $\Pi_0=\Pi^*(P_0)$).
This is the best-case scenario because now Theorem~\ref{thm:expected value} applies directly and the sum term for $t=0$ is correctly calculated.
The issue is that in future time steps, we use action $a_t$ to generate the next estimate $P_1$ and perform the same calculation under the assumption that $\Pi_1 = \Pi^*(P_1)$, which is typically not true.
The true distribution $\Pi_1=M(a_0)\Pi_0$ has a lower entropy, which can be used by the value function to produce a more precise evaluation of the value function.
We term this mechanism causing an inaccurate estimation of the value function \textit{entropy drift,}
So, for all $t\in \{1,2,\dots \}$, it is true that the heuristic may have errors as a result of the true distribution $\Pi_t$ drifting away from $\Pi^*(P_t)$ over time.
The second issue is the assumption that we begin with an entropy maximizing distribution.
For example, in practical application it may often be true that the distribution $\Pi_0\neq \Pi^*(P_0)$, but $\Pi_0\in S(P_0)$, meaning the initial distribution is consistent with the information $P_0$, however, it is not the distribution that maximizes entropy.
We term the mechanism for inaccurate value function estimation \textit{misaligned initial entropy}.

With these notions in mind, we would like to judge in what scenarios the maximum entropy assumption is reasonable.
From the epidemics literature, it has been shown that this assumption is exact in the case of complete networks with the number of nodes taken to infinity \cite{armbruster2017elementary,armbruster2017elementary2}.
Intuitively, this justifies the idea graphs that are highly connected and symmetric graphs are likely stay close to the maximum entropy distributions.
Conversely, we hypothesize that sparsely connected graphs with little symmetry between nodes may deviate more significantly.
In the next subsection, we address this question numerically in terms of entropy drift and misaligned initial entropy.

\subsection{Empirical Evaluation of the Proposed Heuristic}

First, we will formalize our heuristic based on Theorem~\ref{thm:expected value}.
In short, the network operator will use $P_t$ as a running tracker of the security of their computer network given the action sequence $(\vec{a}=a_0,a_1,\dots,a_{T-1})$.
Then, we create a heuristic value function $\hat{V}\approx V$ by making one change to the value function $V$.
We apply Theorem~\ref{thm:expected value} and replace the term in the sum with the approximation $$\sum_{x^j\in X}(\Pi^j_t R(a_t)\odot M(a_t))^T  \vec{1}\approx -\sum_{l\in N}(P_{t+1,l}+ca_{t_l}) \bigg).$$
Doing both of these, we can give our tractable approximation of the value function:
 \begin{equation} \label{eq:approximate value}
     \hat{V}(\vec{a})=\sum^T_{t=0} \bigg( -\gamma^t\sum_{l\in N}(P_{t+1,l}+ca_{t,l}) \bigg).
 \end{equation}

In realistic scenarios, it is desirable to be able to evaluate the value function for any $T\in\mathbb{N}$.
Unfortunately, in the next series of experiments, we want to compare $V$ with $\hat{V}$.
Because $V$ can be extremely slow to compute even for relatively small values of $T<10$, we are limited to considering relatively small values of $T$.
The primary goal of our simulations is to evaluate the quality of the estimates of the value function produced by $\hat{V}_t$, with respect to the two mechanisms causing estimations inaccuracy, entropy drift and misaligned initial entropy.
To do this, we define a special initial distribution:
\begin{equation}
    \Pi(\sigma)=\sigma\Pi'+(1-\sigma)\Pi^*(X\Pi')
\end{equation}
where $\Pi'$ has $\Pi'_i=0.5$ and $\Pi'_j=0.5$ with $x^i=(1,0,1,0,\dots)$ and $x^j=(0,1,0,1,\dots)$.
This distribution $\Pi'$ has a simple interpretation, either all even-indexed nodes are compromised or all odd-indexed nodes are compromised with an equal chance.
The reason we use $\Pi(\sigma)$ as an initial distribution is that it allows us to parameterize the amount of entropy that the initial distribution has, while maintaining the same marginal distributions in $P_0$.
Specifically, for any $\sigma\in [0,1]$ we have $X\Pi(\sigma)=(0.5,0.5,0.5,\dots)=P_0$.
Further, $\Pi(0)=\Pi^*(X\Pi')=(\frac{1}{2^n},\frac{1}{2^n},\frac{1}{2^n},\dots)$ is the uniform distribution which also happens to maximize entropy in this case (matching the form given in \eqref{eq:max entropy}).
Intuitively, entropy measures the uncertainty of a distribution, so it is unsurprising to find that the uniform distribution is the entropy maximizer in this case.
Following that intuition, $\Pi'$ has a relatively low entropy, as there is little uncertainty between only two possible states.
It is straightforward to verify that $\Pi(\sigma)$ decreases in entropy as $\sigma$ increases, meaning that $\sigma$ parametrizes the amount of entropy in the distribution while simultaneously maintaining constant marginal distributions, which means that $\hat{V}$ will provide a constant value for any $\sigma$.
Thus, to evaluate the estimate quality due to entropy drift, we can take $\sigma=0$ giving the case where Theorem~\ref{thm:expected value} is satisfied.
To evaluate the impacts of misaligned initial entropy, we can take $\sigma>0$ and measure how strongly $V$ reacts to distributions with less entropy.

For the first experiment, we seek to address the question: how much does $\hat{V}$ differ from $V$ numerically as a function of entropy?
To do this, we sample $1,000$ controls at random from the control space and evaluate the subtraction $V(\vec{a})-\hat{V}(\vec{a})$ where $V$ is the finite time analog of \eqref{eq:exact value}, for $\sigma\in[0.0,0.25,0.5,0.75,1]$.
For this experiment, we use subgraphs of the graph originally studied in \cite{kazeminajafabadi2024optimal}, where we take the first $n$ nodes (taking attacker as node $0$, and taking all nodes as their notion of ``OR'' nodes), and take $n$ as large as computationally permissive.
Specifically, this graph is defined by $n=8$ \footnote{For brevity, we only include $n=8$, a more comprehensive set of plots, source code and data can be found on \href{https://github.com/descon-uccs/NetworkedAPTDefense}{GitHub}.} resulting in the infectiousness parameters are $\rho_{01}=0.65$, $\rho_{03}=0.6$, $\rho_{12}=0.7$, $\rho_{14}=0.6$, $\rho_{25}=0.6$, $\rho_{36}=0.55$, $\rho_{47}=0.7$, $\rho_{62}=0.7$ (where not listed assume $\rho_{kl}=0$).
For parameters, we use $\alpha=0.1$, $\gamma=0.5$, $c=2.1$, $T=4$.
\begin{figure}[h!]
    \centering
    \includegraphics[scale=0.5]{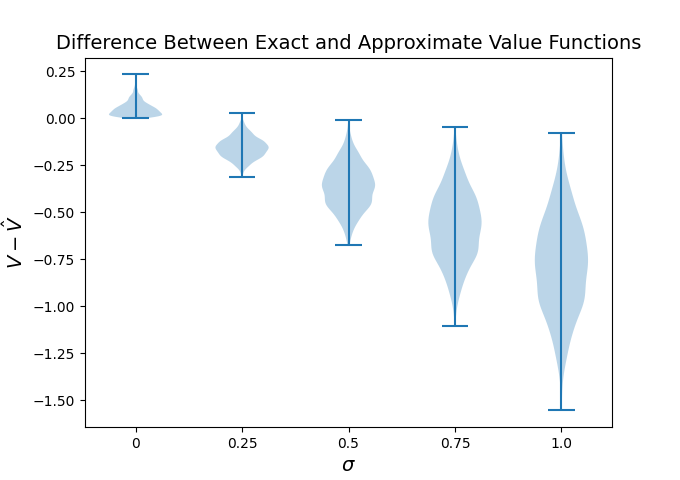}
    \caption{Comparision between the exact value function \eqref{eq:exact value} and the heuristic proposed heuristic value function \eqref{eq:approximate value}.
    \label{fig:PC}}
\end{figure}

The result of this experiment can be seen in Figure~\ref{fig:PC} as a \textit{violin plot} (the thickness of each violin indicates the relative frequency of the value of the $y$-axis in the data).
Each violin plot represents the difference between the exact and approximate values of the respective value functions.
As we move from left to right on the plot, $\sigma$ increases, which corresponds to entropy decreasing.
In the left-most violin, where $\sigma=0$, the assumptions of Theorem 2 are satisfied as $\Pi(0)$ is the entropy maximizing distribution given by \eqref{eq:max entropy}, giving an exact calculation of $V$ at $t=0$.
As previously discussed, the entropy of $\Pi_t$ tends to drift from $P_t$ resulting in some estimation error.
The left violin shows that this error is relatively minor, with the majority of samples grouping up around 0 error.
Interestingly, entropy drift seems to cause the approximate value function to underestimate the true value function, as all samples reported a positive difference.
This is in stark contrast to the remainder of the plots, where we introduce misaligned initial entropy.
Here we see errors much greater in magnitude, which almost universally tend to over predict the value function.
Overall, this plot demonstrates the intuition of the estimator: if the maximum entropy assumption holds in Theorem~\ref{thm:expected value}, then the errors are very small, and the more this assumption is violated in terms of entropy the worse the approximate value function performs.

The previous experiment directly measures how well $\hat{V}$ approximates $V$, however the reported values of a value function are somewhat arbitrary with respect to finding an optimal control.
When solving searching for an optimal control, it could be that a heuristic value function is a poor estimator of the true value function, but if it has the same local minimizers then optimizing it will have the same result as maximizing the true value function.
To this end, we seek to quantify how much our proposed approximate value function $\hat{V}$ distorts the \textit{relative ordering} of the control space compared to the ordering induced by the true value function $V$.
Formally, we define partial order $\geq_V$ such that for two controls $a,a'$, we say $a'\geq_Va$ whenever $V(a')\geq V(a)$.
In this experiment, we compare every pair of control sequences $\vec{a},\vec{a}' \in \{0,1\}^{T\times n}$ and count the occurrence where $\vec{a}'\geq_{V}\vec{a}$ but $\vec{a}'\ngeq_{\hat{V}}\vec{a}$.
We term this situation an \textit{order violation}, and we use the total fraction of order violation divided by the total number of pairs as a measure of how distorted the relative values of the control space are when evaluating with the approximate value function $\hat{V}$.
A relatively small fraction of order violations indicates that the control space is evaluated similarly, and any algorithm that evaluates a trajectory of controls searching for a (possibly local) optimum will be relatively unaffected by the use of an approximate value function.
For this experiment, we use the same parameters as before, except that we generate 100 Erd\H{o}s-R\'enyi graphs with 3 nodes and the probability of each pair of nodes having a directed edge is $0.5$.
Additionally, we set $\rho_{kl}=0.1$ for all edges $(k,l)\in E$.

\begin{figure}
    \centering
    \includegraphics[scale=0.27]{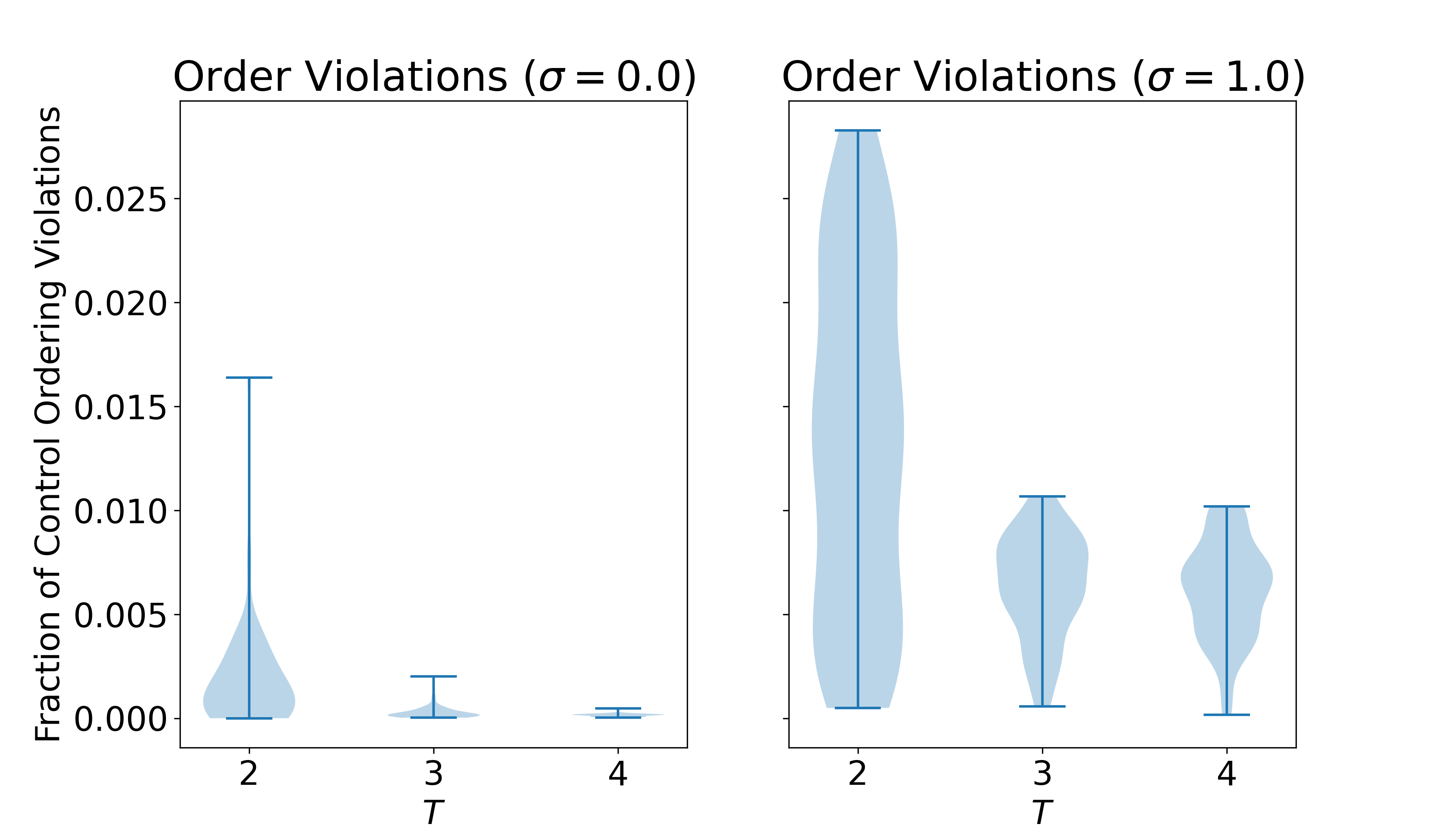}
    \caption{Frequency of Control Order Violations over 100 Erd\H{o}s-R\'enyi graphs. The thickness of a violin plot at a given $y$-axis value indicates the relative frequency that percentage of violations occurred among the sample graphs.}
    \label{fig:COER}
\end{figure}

The results of this experiment can be seen in Figure~\ref{fig:COER}, where we show the results across $T\in \{2,3,4\}$ for $\sigma=0$ (left) and $\sigma=1$ (right).
First, we can see that even in the worst case $\sigma=1$ and $T=2$, at most about 2.5\% of pairs are disordered.
This indicates that it will be relatively unlikely for an algorithm that searches over controls and compares them will encounter disordered pairs of controls frequently.
Second, we can see two trends in terms of $T$ and $\sigma$.
As we increase $T$, the number of controls pairs appears to grow faster than the number of disordered pairs.
Second, just like in the previous example, setting $\sigma=1$ significantly degrades the estimation quality which reflects the entropy assumption being violated.


\section{Conclusion}
In this work, we propose a tractable heuristic value function with strong theoretical connections to the emergent value function.
Numerically, we showed that our heuristic performs very well when the entropy assumptions it is based on are true.
However, it does see some performance degradation as those assumptions are violated.
For future work, we would like to investigate heuristic and tractable control approaches that take advantage of the heuristic value function posed in this work.
Due to the exponential size of the control space, future work should study methods that tractably search the control space.

\bibliographystyle{ieeetr}
\bibliography{references}

@inproceedings{braga2011optimal,
  title={Optimal state estimation for Boolean dynamical systems},
  author={Braga-Neto, Ulisses},
  booktitle={Asilomar Conference on Signals, Systems and Computers},
  year={2011},
  organization={IEEE}
}

@article{XuTDSC2012,
  author    = {S. Xu and W. Lu and Z. Zhan},
  title     = {A Stochastic Model of Multivirus Dynamics},
  journal   = {IEEE Transactions on Dependable and Secure Computing},
  volume    = {9},
  number    = {1},
  year      = {2012},
  pages     = {30-45},
}

@inproceedings{kazeminajafabadi2024optimal,
  title={Optimal detection for Bayesian attack graphs under uncertainty in monitoring and reimaging},
  author={Kazeminajafabadi, Armita and Ghoreishi, Seyede Fatemeh and Imani, Mahdi},
  booktitle={2023 American Control Conference (ACC)},
  year={2024}
}

@article{paarporn2017networked,
  title={Networked SIS epidemics with awareness},
  author={Paarporn, Keith and Eksin, Ceyhun and Weitz, Joshua S and Shamma, Jeff S},
  journal={IEEE Transactions on Computational Social Systems},
  volume={4},
  number={3},
  pages={93--103},
  year={2017},
  publisher={IEEE}
}

@article{kazeminajafabadi2024optimal_journal,
  title={Optimal Joint Defense and Monitoring for Networks Security under Uncertainty: A POMDP-Based Approach},
  author={Kazeminajafabadi, Armita and Imani, Mahdi},
  journal={IET Information Security},
  volume={2024},
  number={1},
  pages={7966713},
  year={2024},
  publisher={Wiley Online Library}
}

@article{armbruster2017elementary,
  title={Elementary proof of convergence to the mean-field model for the SIR process},
  author={Armbruster, Benjamin and Beck, Ekkehard},
  journal={Journal of mathematical biology},
  volume={75},
  pages={327--339},
  year={2017},
  publisher={Springer}
}

@article{armbruster2017elementary2,
  title={An elementary proof of convergence to the mean-field equations for an epidemic model},
  author={Armbruster, Benjamin and Beck, Ekkehard},
  journal={IMA Journal of Applied Mathematics},
  volume={82},
  number={1},
  pages={152--157},
  year={2017},
  publisher={Oxford University Press}
}

@article{nowzari2016analysis,
  title={Analysis and control of epidemics: A survey of spreading processes on complex networks},
  author={Nowzari, Cameron and Preciado, Victor M and Pappas, George J},
  journal={IEEE Control Systems Magazine},
  volume={36},
  number={1},
  pages={26--46},
  year={2016},
  publisher={IEEE}
}

@article{pare2020modeling,
  title={Modeling, estimation, and analysis of epidemics over networks: An overview},
  author={Par{\'e}, Philip E and Beck, Carolyn L and Ba{\c{s}}ar, Tamer},
  journal={Annual Reviews in Control},
  volume={50},
  pages={345--360},
  year={2020},
  publisher={Elsevier}
}

@incollection{XuBookChapterCD2019,
author={Shouhuai Xu},
title={Cybersecurity Dynamics: A Foundation for the Science of Cybersecurity},
booktitle={Proactive and Dynamic Network Defense},
publisher={Springer},
volume={74},
year={2019},
pages="1--31",
}

@INPROCEEDINGS{XuMTD2020, 
author={S. Xu}, 
booktitle={ACM Workshop on Moving Target Defense}, 
title={The Cybersecurity Dynamics Way of Thinking and Landscape (invited paper)}, 
year={2020}, 
}

@inproceedings{collins2025efficient,
  title={Efficient state estimation of a networked flipit model},
  author={Collins, Brandon and Gherna, Thomas and Paarporn, Keith and Xu, Shouhuai and Brown, Philip N},
  booktitle={2025 IEEE 64th Conference on Decision and Control (CDC)},
  pages={6808--6813},
  year={2025},
  organization={IEEE}
}

@article{al2024mitre,
  title={MITRE ATT\&CK: State of the art and way forward},
  author={Al-Sada, Bader and Sadighian, Alireza and Oligeri, Gabriele},
  journal={ACM Computing Surveys},
  volume={57},
  number={1},
  pages={1--37},
  year={2024},
  publisher={ACM New York, NY}
}

@article{senatereport2014,
  author      = {U.S. Senate-Committee on Commerce Science and Transportation},
  title       = {``Kill Chain'' Analysis of the 2013 Target Data Breach},
  institution = {United States Senate},
  address     = {Washington, DC},
  year        = {2014},
  url         = {https://www.public.navy.mil/spawar/Press/Documents/Publications/03.26.15_USSenate.pdf}
}

@article{papadimitriou1987complexity,
  title={The complexity of Markov decision processes},
  author={Papadimitriou, Christos H and Tsitsiklis, John N},
  journal={Mathematics of operations research},
  volume={12},
  number={3},
  pages={441--450},
  year={1987},
  publisher={INFORMS}
}

@article{hauskrecht2000value,
  title={Value-function approximations for partially observable Markov decision processes},
  author={Hauskrecht, Milos},
  journal={Journal of artificial intelligence research},
  volume={13},
  pages={33--94},
  year={2000}
}

\end{document}